\documentclass[12pt]{article}
\usepackage{amscd,amsmath,amsopn,amssymb,amsthm,multicol,xcolor,ulem}
\usepackage{authblk}

\newcommand\com[1]{}
\newcommand\killing{{\mathfrak{k}}}
\newcommand\generators{{\mathfrak{k}_{\mathcal{H}}}}
\newcommand\generatorsdeg{{\mathfrak{k}_{\mathcal{H}}^{\mathrm{deg}}}}

\newcommand\E{\mathcal{E}}
\newcommand\g{{\mathfrak g}}

\theoremstyle{plain}
\newtheorem{theorem}{Theorem}
\newtheorem{prop}{Proposition}
\newtheorem{cor}{Corollary}
\newtheorem{lemma}{Lemma}
\theoremstyle{definition}
\newtheorem{definition}{Definition}
\newtheorem{example}{Example}
\newtheorem{remark}{Remark}

\newtheorem{manualtheoreminner}{Theorem}
\newenvironment{manualtheorem}[1]{%
  \renewcommand\themanualtheoreminner{#1}%
  \manualtheoreminner
}{\endmanualtheoreminner}

\let\isout\sout \renewcommand{\sout}[1]{\ifmmode\text{\isout{\ensuremath{#1}}}\else\isout{#1}\fi}

\definecolor{colorTM}{rgb}{.2,.7,.2}

\begin{document}

\title{Detecting horizons of symmetric black holes using relative differential invariants}

\author[a]{David McNutt}
\author[b]{Eivind Schneider}

\affil[a]{\footnotesize Center for Theoretical Physics, Polish Academy of Sciences, Warsaw, Poland}
\affil[b]{\footnotesize Department of Mathematics and Statistics, UiT: The Arctic University of Norway, Tromsø, Norway}

\maketitle

\begin{abstract}
Let $\killing$ be a nontrivial finite-dimensional Lie algebra of vector fields on a manifold $M$, and consider the family of Lorentzian metrics on $M$ whose Killing algebra contains $\killing$. We show that scalar relative differential invariants, with respect to a Lie algebra of vector fields on $M$ preserving $\killing$, can be used to detect the horizons of several well-known black holes.  In particular, using the Lie algebra structure of $\killing$, we construct a general relative differential invariant of order $0$ that always vanishes on $\killing$-invariant Killing horizons. 
\end{abstract}

\tableofcontents

\section{Introduction}

Black holes are solutions to Einstein's field equations describing the result of the gravitational collapse of stellar objects. The study of these solutions gives insight into higher curvature regimes in the universe and by studying the boundaries of these surfaces, it is possible to model the possible gravitational waves arising from perturbed solutions \cite{Thornburg,Bishop}. However, to do this, it is necessary to have a firm definition of a black hole horizon. In the case of stationary black holes used in astrophysics, such as the Kerr solution \cite{kerr1963gravitational}, the horizon can be determined and this coincides for the event horizon. For more general black holes that admit a time-like symmetry, determining a quasi-local invariant surface is a difficult problem.

The event horizon is not an ideal candidate as a boundary for a black hole. It is, generally, not quasi-local and in black hole solutions admitting a positive cosmological constant may not exist at all \cite{senovilla2023ultra}. In practice, quasi-local surfaces, known as future trapped outer horizons or more briefly as apparent horizons are used to determined boundaries of dynamical black holes \cite{Thornburg}. However, such surfaces are often not invariant in the sense that they depend on the foliation of spacetime and therefore cannot be considered fully physical. The observation that in the case of stationary and weakly isolated black hole solutions \cite{Ashtekar} the horizon can be detected as a quasi-local invariant surface has motivated the Geometric horizon conjectures which posit an appropriate quasi-local invariant hypersurface that bounds the black hole must arise as the zero-set of some scalar curvature invariant \cite{ADA}.

The classification of black hole solutions has been investigated using scalar polynomial curvature invariants  \cite{carminati1991algebraic} and the horizon has been shown to be detectable in terms of these curvature invariants \cite{AbdelqaderLake2015, PageShoom2015}. More generally, the classification of black hole solutions can be accomplished using Cartan invariants \cite[Chapter 9]{kramer} and the Killing horizons were shown to be zero-sets of certain Cartan invariants \cite{brooks2018cartan}. 

Initially the investigation of horizon detecting curvature invariants presupposed knowledge of the horizon's location for a given spacetime and a curvature invariant was found that vanishes on the horizon. Using symmetry arguments, \cite{PageShoom2015} described the construction of scalar polynomial curvature invariant that would detect the Killing horizon. However, the existence of such curvature invariants was not guaranteed. Later, by employing the framework of weakly isolated horizons \cite{Ashtekar} a first order curvature invariant was found that would always detect the Killing horizon, without knowledge of its location \cite{ColeyMcNutt}. 

There are other approaches to classification of spacetimes, and in particular black hole solutions, beyond curvature invariants. In principle, an IDEAL classification of spherically symmetric black hole solutions is possible \cite{ferrando2009intrinsic}. However, a general approach to the IDEAL classification has yet to be determined for general spacetimes. Lastly, one can obtain differential invariants by directly investigating the appropriate Lie pseudogroup actions on jet bundles. For instance, this approach has been used to classify Kundt spacetimes \cite{kruglikov2021differential} and also spacetimes admitting two commuting Killing vectors \cite{marvan2008local, ferraioli2020equivalence}. Kerr-like solutions can be treated within the framework of \cite{ferraioli2020equivalence}, but the differential invariants and analysis presented there are not adapted to horizon detection, as the required differential invariant that detects the Killing horizon must be extracted from an algebraic combination of the differential invariants provided in \cite{ferraioli2020equivalence}. 

In this paper we will focus on spacetimes admitting Killing vectors, with the goal of determining relative differential invariants that vanish on physically important hypersurfaces, such as black hole horizons and Killing horizons. Let $\killing$ be a given Lie algebra of of vector fields on a manifold $M$ and consider the family of metrics on $M$ with Killing algebras containing $\killing$. We fix the coordinate expressions of the Killing vectors, imagining that observers agree on the local form of the Lie algebra of Killing vectors (but not necessarily on a basis of the Lie algebra). By investigating the remaining coordinate freedom, and in particular the corresponding scalar differential invariants, we aim to find conditions that determine important hypersurfaces in $M$, and in particular black hole horizons. 

Our approach is as follows. We consider a family of metrics with Killing algebra containing $\killing$. They are sections of the bundle $\pi \colon S_{\text{Lor}}^2 T^*M \to M$ 
that satisfy a partial differential equation (system), (PDE), which we will in general denote by $\E$. In the extreme case, $\E$ is defined by the system
\begin{equation}
L_X g = 0, \qquad \forall X \in \killing, \label{eq:pde}
\end{equation}
in which case the space $\mathrm{Sol}(\E)$ of solutions consists of all metrics for which every element in $\killing$ is a Killing vector field. In general, we will also allow for $\E$ to be a sub-PDE of that defined by \eqref{eq:pde}.  We then treat the PDE $\E$ geometrically, namely as a family of submanifolds in the spaces of jets of sections of the bundle $\pi$, i.e., $\E^i \subset J^i \pi$ for $i\in\{0,1,2,\dots\}$.

Let $G\subset \mathrm{Diff}_{\mathrm{loc}}(M)$ be the Lie pseudogroup of symmetries of $\E$ that also preserve the Lie algebra $\killing$ of vector fields on $M$. It encodes the remaining coordinate freedom, after fixing $\killing$ and $\E$. The Lie algebra (sheaf) $\g$  of vector fields corresponding to the Lie pseudogroup $G$ contains $\killing$ as an ideal, and it is always a Lie subalgebra of the Lie algebra of all vector fields preserving $\killing$,
\begin{equation}
\{X \in \mathcal{D}(M) \mid [X,K] \in \killing, \forall K \in \killing\}. \label{eq:sym}
\end{equation}
In this paper, we will in general not impose the Einstein equation on our Lorentzian manifolds, but simply note that it is possible to do that within our framework. Due to covariance of the Einstein equations, the sub-PDE of \eqref{eq:pde} obtained by imposing the Einstein equation has the full Lie algebra \eqref{eq:sym} as symmetries. The Lie algebra $\g$ prolongs to a Lie algebra $\g^{(i)}$ of vector fields on $\E^i$ for $i=0,1,2,\dots$.

Our goal is to investigate to which extent scalar relative $\g^{(i)}$-invariants on $\E^i$ (or relative differential invariants of order $i$) can be used to detect black hole horizons of metrics admitting Killing vectors. It is natural to conjecture that they would do that, at least in some cases, as they determine hypersurfaces in a coordinate-independent way. We outline the structure and main results of the paper. 

In Section \ref{sect:R} we show how an ideal of $\g$ can be used to construct one relative differential invariant of order $0$, and also one of order $1$ giving the theorem: 
\begin{manualtheorem}{1}
Let $\mathfrak{i} = \langle K_1, \dots, K_r\rangle$ be an ideal of the Lie algebra $\g$ of symmetries of the PDE $\E$, and assume that $\dim \mathfrak{i}=r \geq 1$. Then 
\[ L_{X} (K_1 \wedge \cdots \wedge K_r) = \lambda_X (K_1 \wedge \cdots \wedge K_r), \qquad \forall X \in \g\]
where $\lambda \in \g^*$. Consequently, the function 
\[R^{\mathfrak{i}}= \| K_1^{(0)} \wedge \cdots \wedge K_r^{(0)}\|_{h}^2\] on $\E^0 \subset J^0 \pi$ is a relative invariant with weight $2\lambda$. Moreover, the function $S^{\mathfrak{i}}=\|\bar{d} R^{\mathfrak{i}}\|^2_{h}$ on $\E^1 \subset J^1 \pi$ is a relative invariant with weight $4 \lambda$. 
\end{manualtheorem}
\noindent We note that for some ideals $R^{\mathfrak{i}}$ may be an absolute invariant with $\lambda=0$, and it is even possible that $R^{\mathfrak{i}}$ is simply a constant function. 

We will focus on $\g$-invariant ideals of $\killing$ (which are also ideals of $\g$). Depending on the initial Lie algebra $\killing$ there may be several ways of choosing the ideal $\mathfrak{i}$, and in general different choices will lead to different relative differential invariants. Due to the importance of $\mathfrak{i}$ being an ideal of $\g$, special attention is paid to characteristic ideals of $\killing$. In particular, the radical ideal $\mathfrak{r}$ of $\killing$, and the elements of its derived sequence $\mathfrak{r}_1, \mathfrak{r}_2, \dots, \mathfrak{r}_k$ are characteristic ideals. In Section \ref{sect:abelianideal} we use the terms of the derived sequence of $\mathfrak{r}$ to construct a characteristic abelian ideal $\mathfrak{a}(\mathfrak{r})$, and show that it intersects nontrivially with any other ideal of $\mathfrak{r}$.  

In Section \ref{sect:killingintro} we introduce the concept of $\killing$-invariant Killing horizon, and show that the existence of such implies that $\killing$ contains a nontrivial solvable ideal, spanned by the generators of the Killing horizon. Thus, the radical $\mathfrak{r}$ in the Levi decomposition $\killing = \mathfrak{s} \ltimes \mathfrak{r}$ is nontrivial, which in turn guarantees that $\mathfrak{a}(\mathfrak{r})$ is nontrivial. We show that the relative invariant constructed from $\mathfrak{a}(\mathfrak{r})$ by Theorem \ref{th:R} always vanishes on $\killing$-invariant Killing horizons. 
\begin{manualtheorem}{2}
Let $(M,g)$ be a Lorentzian manifold with Killing algebra $\killing = \mathfrak{s} \ltimes \mathfrak{r}$, let $\g$ be a Lie algebra of vector fields on $M$ containing $\killing$ as an ideal, and let $\mathfrak{a}(\mathfrak{r}) = \langle K_1,\dots, K_r\rangle$ be the abelian ideal defined by \eqref{eq:a(r)}. If $\mathcal{H}$ is a $\killing$-invariant Killing horizon, then the function 
    \[ R_g = \| K_1 \wedge \cdots \wedge K_r \|_g^2\]
    vanishes on $\mathcal{H}$. The function $R_g$ is the restriction of the relative invariant of Theorem \ref{th:R} to $g$: $R_g = R^{\mathfrak{a}(\mathfrak{r})} \circ g$. 
\end{manualtheorem}
 While it is possible to determine the Killing horizons of a stationary black hole by computing the norm of the appropriate Killing vector fields which are null on a particular horizon and are proportional to the generator for the hypersurface, in practice this requires finding a Killing vector field for the horizon. This difficulty is exemplified in the example of the Kerr black hole where determining the Killing vector field for the outer horizon is not completely trivial \cite[section 5.2]{poisson2004}. On the other hand, computing $R_g$ for the Kerr black hole is easy and requires no a priori knowledge about the Killing horizon, as shown in Section \ref{sect:R2}.

In Section \ref{sect:examples} we focus on several specific examples of Killing algebras. We apply the results of Section \ref{sect:idealsandinvariants} together with an independent orbit analysis on $\E^0$ to get a complete picture of the relative differential invariants of order $0$, and in some cases we also make a similar analysis on $\E^1$.  When $\killing$ is the Killing algebras of the Schwarzschild metric, the Kerr metric, or the Near Horizon geometries we see that the well-known horizons can be deduced from the orbit structure on $\E^0$. Furthermore, since these horizons are $\killing$-invariant Killing horizons, they are also detected by the relative differential invariant of Theorem \ref{th:invariantH}. 

Even in some cases where the relative differential invariant of Theorem \ref{th:invariantH} is not defined (for example if $\killing$ is simple), the orbit structure on $\E^i$ can be used to detect horizons. In Section \ref{sect:spherical} we consider the family of spherically symmetric spacetimes, for which $\killing$ contains no nontrivial ideals. In this case, $\E^0$ is 8-dimensional and all orbits on $\E^0$ are 7-dimensional. In particular, there are no proper relative invariants on $\E^0$ (only an absolute invariant). However, there are relative invariants on $\E^1$, and we show that one of these determines the horizon of imploding spherically symmetric metrics. 

To provide background for these results, we begin in Section \ref{sect:prerequisites} with a quick introduction to jets and differential invariants, after which we outline the general set-up of our approach in more detail.

\section{Jets and differential invariants} \label{sect:prerequisites}
In this section we give a brief introduction to some language and notations related to jet spaces and differential invariants. For more information about these subjects, we refer to \cite{olver1995equivalence, kruglikovlychagin2007handbook, kruglikovlychagin2016global}.  

\subsection{Jet spaces} \label{sect:jets}
A Lorentzian metric on a 4-dimensional manifold $M$ is a section of the bundle $\pi \colon S^2_{\mathrm{Lor}} T^*M \to M$ of symmetric 2-forms with Lorentzian signature. A metric $g \in \Gamma(\pi)$ defined on a domain $U \subset M$ determines a submanifold $g(U) \subset S^2_{\mathrm{Lor}} T^*M$. We say that two metrics $g,\tilde g$ are $k$-equivalent at $a \in M$ in a neighborhood $U$ of $x$ if the submanifolds $g(U), \tilde g(U) \subset S^2_{\mathrm{Lor}} T^*M$ are tangent up to order $k$ at $a$, for $k=0,1,2,\dots$. We call the equivalence class of $g$ the $k$-jet of $g$ at $a$, and denote it by $[g]_a^k$. It can be thought of as the collection of $k$th-degree Taylor polynomials of the components of $g$ at the point $a$. We denote by $J_a^k \pi$ the space of $k$-jets of metrics at the point $a$, and define the jet space $J^k \pi = \bigsqcup_{a \in M} J_a^k \pi$. Note that $J^0 \pi =S^2_{\mathrm{Lor}} T^*M$. We have the natural projections $\pi_{k,l} \colon J^k \pi \to J^l \pi$ for $0 \leq l < k$ and $\pi_k \colon J^k \pi \to M$. 

A local diffeomorphism $\varphi$, defined on $U \subset M$, prolongs to a local diffeomorphism defined on $\pi_k^{-1}(U) \subset J^k \pi$ in the following way. Consider a general point $\theta = [g]_a^k \in \pi_k^{-1}(U) \subset J^k \pi$, where $g$ is a section of $\pi$ defined in a neighborhood of $a \in U$. Then $\varphi^{(k)}(\theta)= [(\varphi^{-1})^* g]_{\varphi(a)}^k$. In a similar way, a vector field $X$ on $U \subset M$ can be uniquely prolonged to a vector field $X^{(k)}$ on $\pi_k^{-1}(U) \subset J^k \pi$. 

The total space $S^2_{\mathrm{Lor}} T^*M$ is naturally equipped with a horizontal symmetric 2-form $h$, defined by $h_{\theta} = \theta \in S_{\mathrm{Lor}}^2T^*_a M$ for any point $\theta \in S_{\mathrm{Lor}}^2 T_a^* M$. In particular, for $g \in \Gamma(\pi)$ we have $h \circ g = g$. The 2-form $h$ is invariant under the prolongation of any local diffeomorphism $\varphi$ and any vector field $X$ on $M$, meaning 
\[ (\varphi^{(0)})^* h = h, \qquad L_{X^{(0)}} h = 0. \]
Often we use these properties for computing formulas for $\varphi^{(0)}$ and $X^{(0)}$ from the formulas for $\varphi$ and $X$, respectively. 

For example, if $x^1, \dots, x^n$ are local coordinates on $M$, they can be extended to a coordinate system on   $S^2_{\mathrm{Lor}} T^*M$ by adding coordinates $u_{ij}$, $1 \leq i <j \leq n$ in such a way that $h$ takes the form 
\[ h= \sum_{i \leq j} u_{ij} dx^i dx^j.\]
Furthermore, they can be extended canonically to coordinates $(u_{ij})_\sigma$ on $J^k \pi$, where $\sigma$ is a multi-index with $|\sigma|\leq k$. 

If $g = \sum_{i\leq j} g_{ij}(x) dx^i dx^j$ is a section of $\pi$, we write 
$u_{ij}|_g = u_{ij} \circ g= g_{ij}(x)$, and easily verify that $h \circ g = g$. The lift of a general vector field $X= a^i(x) \partial_{x^i}$ takes the form 
\[ X^{(0)} = \sum_{i} a^i(x) \partial_{x^i} + \sum_{i\leq j} b_{ij}(x,u) \partial_{u_{ij}},\]
where the functions $b_{ij}(x,u)$ are uniquely determined from the linear algebraic system $L_{X^{(0)}} h =0$. Higher prolongations $X^{(k)}$ are then computed by the standard jet-prolongation formulas (see for example \cite[sect.~1.5]{kruglikovlychagin2007handbook}). 

A function $f$ on $J^k \pi$ can be composed with the $k$-jet of a section $g$ of $\pi$ to give a function $f_g := f \circ j^k g$ on the base $M$. The horizontal exterior derivative, denoted by $\bar d$, satisfies $(\bar d f)\circ j^{k+1} g = d(f\circ j^k g)$ for any $g$, where $d$ is the exterior derivative on $M$. In coordinates it can be written as $\bar df = D_{x^{l}}(f) dx^l$, where $D_{x^l}(f)$ denotes the total derivative operator 
\[ D_{x^{l}} = \partial_{x^l} + \sum_{i \leq j} \sum_{\sigma}(u_{ij})_{\sigma+1_l} \partial_{(u_{ij})_\sigma}.\]
The total derivative satisfies $D_{x^l}(f)\circ j^{k+1} g = \partial_{x^l}(f \circ j^k g)$. 

\subsection{PDE coming from Killing vectors} \label{sect:pde}
In this paper, $\killing$ will be a fixed finite-dimensional Lie algebra of vector fields on $M$.  The elements of $\killing$ are Killing vectors of a section $g \in \Gamma(\pi)$ if and only if 
\begin{equation} L_K g = 0, \qquad \forall K \in \killing. \label{eq:killing}
\end{equation}
This gives an (in general overdetermined) system $\mathcal E$ of PDEs on the components of $g$. Bringing this system to involution, which we will always do, gives us for each $k \geq 0$ a submanifold $\E^k \subset J^k \pi$ such that the maps $\pi_{i+1,i}|_{\E^i} \colon \E^{i+1} \to \E^{i}$ are submersions. In this case, $\E$ is called a formally integrable PDE \cite{kruglikovlychagin2007handbook}. It is possible that the inclusion $\E^0 \subset S^2_{\mathrm{Lor}} T^*M$ is strict. Let $\tau^i \colon \E^i \to J^i \pi$ be the inclusion map. Then $h_\E=(\tau^0)^* h$ is a horizontal symmetric 2-form on $\E^0$.

We will use $\E$ to refer to the infinite collection $\{\E^0, \E^1, \dots \}$ of submanifolds. We denote the space of (smooth) solutions of $\mathcal E$ by $\mathrm{Sol}(\mathcal E)$:
\[\mathrm{Sol}(\mathcal E) = \{g \in \Gamma(S^2_{\text{Lor}} T^*M) \mid L_K g = 0, \forall K \in \killing\}.\]
Geometrically, the fact that $g \in \mathrm{Sol}(\E)$ means that $[g]_a^k \in \E^k$. 

We will throughout this paper assume that $\killing$ is the Lie algebra of Killing vectors for at least one Lorentzian metric on $M$, so that $\mathrm{Sol}(\E)$ is nonempty. If $\killing$ is the Lie algebra of killing vectors of a section $g \in \Gamma(\pi)$, we call it the Killing algebra of $g$. 

In some cases it is desirable to consider a sub-PDE $\tilde \E$ given by submanifolds $\tilde \E^i \subset \E^i$, $i=0,1,2, \dots$. For example, we may add to \eqref{eq:killing} the constraints of the Einstein equation.

\subsection{The Lie algebra preserving $\killing$}
There exists a Lie pseudogroup $G$ consisting of the diffeomorphisms on $M$ preserving $\killing$:
\[ G= \{\varphi \in \mathrm{Diff}_{\mathrm{loc}}(M) \mid \varphi_* \killing = \killing\} .\] 
We denote the corresponding Lie algebra (sheaf) of vector fields by $\g$: 
 \[ \g = \{X \in \mathcal{D}(M) \mid [X,K] \in \killing, \forall K \in \killing\}.\]
It is obvious that $\killing$ is an ideal of $\g$. Both $G$ and $\g$ are often infinite-dimensional. 

Any diffeomorphism in $G$ lifts to a diffeomorphism of $J^k \pi$ for any integer $k \geq 0$, and it is a symmetry of $\mathcal E$, i.e., it preserves the set $\mathcal E^k$ for each $k$, and also the space of solutions $\mathrm{Sol}(\E)$.  In most cases we will work with $\g$ rather than $G$. The Lie algebra $\g$ consists of (infinitesimal) symmetries of $\E$:
\[ X^{(k)}_\theta \in T_\theta \E^k, \qquad \forall \theta \in \E^k, \forall X \in \g.\]

As mentioned in the previous subsection, it is possible to restrict to a sub-PDE $\tilde \E$. If $\tilde \E$ is not $\g$-invariant, one should restrict to a subalgebra $\tilde \g \subset \g$ of symmetries of $\tilde \E$. 
We will use the notation 
\[ \g^{(i)} = \{ X^{(i)} \mid X \in \g\} \]
for the Lie algebra of prolonged vector fields. 

\subsection{Differential invariants}
Let $\E$ be a PDE, defined as a sequence of submanifolds $\E^i \subset J^i \pi$, $i=0,1,2,\dots$, and let $\g$ a Lie algebra of symmetries of $\E$. We define absolute and relative differential invariants in this context.

\begin{definition} \label{def:invariants}
A function $I$ on $\E^k$ is called an absolute differential invariant of order $k$ if
\[ X^{(k)}(I) = 0, \qquad \forall X \in  \g. \] 
A function $R$ on $\E^k$ is called a relative differential invariant of order $k$ if
\[ X^{(k)}(R)= \lambda_X R, \qquad \forall X \in  \g,\]
for some $\lambda  \in \mathrm{Hom}(\mathfrak g, C^\infty(J^k \pi))$. 
\end{definition}
Both equations in the definitions are required to hold only on $\E^k$ (in general, differential invariants on $\E^k$ do not extend to differential invariants on $J^k \pi$). All the  differential invariants that we consider will turn out to be rational. We will not give a detailed explanation of this fact, but simply note that it is related to transitivity of $\g$ and algebraicity of the corresponding Lie pseudogroup (see \cite{kruglikovlychagin2016global} for more details). Moreover, for polynomial relative invariants we actually have $\lambda \in \mathrm{Hom}(\g, C^\infty(M))$, as explained in \cite{kruglikovschneider2024invariantdivisors}.

Absolute differential invariants are constant on $\g^{(k)}$-orbits, meaning that their level sets are $\g^{(k)}$-invariant. For relative invariants, only the zero sets are invariant in general. Any nonvanishing function on $\E^k$ technically satisfies the conditions of Definition \ref{def:invariants}, and in general we will avoid these. (The space of relative invariants forms a group under multiplication. This group can be endowed with an equivalence relation that identifies relative invariants that differ by a nonvanishing factor, under which any nonvanishing function is equivalent to a constant function. See \cite{kruglikovschneider2024invariantdivisors} and references therein for more details.)
\begin{example}
 The function $\det(h)$ on $\E^0 \subset J^0 \pi$ satisfies the condition \[X^{(0)}(\det(h))= \lambda_X \det(h), \qquad \forall X \in \g,\] but $\det(h)$ never vanishes on $\E^0 \subset J^0 \pi= S_{\mathrm{Lor}}^2 T^*M$. We therefore do not consider it to be a proper relative differential invariant.  
\end{example}

A differential invariant is in general not a function on $M$, but on $\E^k \subset J^k \pi$. Let $I$ be an absolute differential invariant of order $k$. Then, for any $g \in \mathrm{Sol}(\E)$ the function $I|_g=I \circ j^k g$ is a function on $M$. Intuitively, its invariance means the following: Its definition in terms of the components of $g$ and their partial derivatives up to order $k$ does not depend on the coordinates in which $g$ was expressed, with the coordinate freedom being given by the (Zariski connected component of the) Lie pseudogroup $G$ corresponding to $\g$ \cite[Section 1.1]{kruglikov2019differential}. Assume that $I|_g$ is not constant in any neighborhood of  $a \in M$. Then we can, in a neighborhood of $a$ define a hypersurface $I|_g=C$ in a $\g$-invariant way. Here $C$ can be any constant in the image of $I|_g$, and by varying $C$ one easily realizes that $M$ is foliated by such hypersurfaces in a neighborhood of $a$. For a relative invariant $R$ of order $k$, we define the restriction $R|_g=R \circ j^k g$ in the same way. If we assume that $R|_g$ does not vanish entirely on $M$, and that it does vanish on some points, we will often get a unique invariantly defined hypersurface $\{R|_g=0\} \subset M$. 
\begin{remark}
    In general, the relative invariants computed in this paper can not be restricted to arbitrary metrics, but only to those satisfying the PDE $\E$. This is explained by the fact that functions on $\E^k \subset J^k \pi$ can not necessarily be extended to invariant functions on $J^k \pi$. Invariants of this type are sometimes referred to as ``conditional invariants'' (see \cite{kruglikovschneider2024invariantdivisors}).
\end{remark}

Since relative differential invariants single out particular hypersurfaces, rather than providing a foliation of such, they will be our main object of study. However, note that for a hypersurface to be ``special'' it is strictly speaking not sufficient that it is given in terms of a relative differential invariant, as the following argument shows. Let $I=R/Q$ be a rational absolute differential invariant, and assume that the level sets of $I|_g$ foliates the spacetime. Any leaf of this foliation is given by $I|_g=C$ for some constant $C$. This equation is equivalent to $R|_g-C Q|_g =0$, where $R-C Q$ is a relative invariant.

Given a Lie algebra $\g$ of vector fields on $M$, and its prolongation $\g^{(k)}$ to $\E^k$, the important relative invariants of order $k$ are often found by locating the points in $\E^k$ where the orbit dimension drops. In particular, if $\g^{(k)}$ is transitive almost everywhere on $\E^k$, then all proper relative invariants are found in this way. 

If generic $\g^{(k)}$-orbits have smaller dimension than $\E^k$, it is possible (but not necessary) that the orbit dimension drops on the level set of an absolute invariant. To illustrate this, let us consider the following two very simple examples of Lie algebras on $\mathbb R^2$: $\g_1=\langle y \partial_x \rangle$ and $\g_2 = \langle x \partial_x \rangle$. In both cases $I=y$ is an absolute invariant. In the first case $\{I=0\}$ is exactly the hypersurface consisting of points where the orbit dimension drops. In the second case $R=x$ is an additional relative invariant, and $\{R=0\}$ is the hypersurface of singular points. These examples illustrate general phenomena that occur for differential invariants as well.

\section{Differential invariants, and ideals of $\killing$} \label{sect:idealsandinvariants}
In this section we use the structure of the Killing algebra $\killing$ to construct relative differential invariants with respect to a Lie algebra $\g$ of vector fields on $M$ that contains $\killing$ as an ideal. In Section \ref{sect:R} we show how to construct a general relative differential invariant of order 0 (and one of order 1) from an ideal of $\g$. We discuss candidates for such ideals in Section \ref{sect:abelianideal}. In particular, we construct an abelian characteristic ideal $\mathfrak{a}(\mathfrak{r})$ from the derived sequence of the radical ideal $\mathfrak{r}$ of $\killing$ and show that it intersects nontrivially with all ideals of $\mathfrak{r}$. In Section \ref{sect:killingintro} we show that the (solvable) Lie algebra $\generators$ of generators of a Killing horizon $\mathcal{H}$ is an ideal in $\killing$ if $\mathcal{H}$ is $\killing$-invariant. The fact that $\mathfrak{a}(\mathfrak{r}) \cap \generators \neq \{0\}$ implies that the relative differential invariant given by the ideal $\mathfrak{a}(\mathfrak{r})$ vanishes on $\mathcal{H}$.

\subsection{Ideals of $\g$ giving rise to relative invariants}\label{sect:R}
As in Section \ref{sect:prerequisites}, let $\killing$ be a finite-dimensional Lie algebra, $\g$ a larger Lie algebra that contains $\killing$ as an ideal, and let $\E$ be $\g$-invariant PDE system on the space of $\killing$-invariant Lorentzian metrics, meaning that the vector fields of  $\g^{(i)}$ are tangent to $\E^{i} \subset J^i \pi$. 

Our first theorem does not involve $\killing$ explicitly. 
\begin{theorem} \label{th:R}
Let $\mathfrak{i} = \langle K_1, \dots, K_r\rangle$ be an ideal of the Lie algebra $\g$ of symmetries of the PDE $\E$, and assume that $\dim \mathfrak{i}=r \geq 1$. Then 
\[ L_{X} (K_1 \wedge \cdots \wedge K_r) = \lambda_X (K_1 \wedge \cdots \wedge K_r), \qquad \forall X \in \g\]
where $\lambda \in \g^*$. Consequently, the function 
\[R^{\mathfrak{i}}= \| K_1^{(0)} \wedge \cdots \wedge K_r^{(0)}\|_{h}^2\] on $\E^0 \subset J^0 \pi$ is a relative invariant with weight $2\lambda$. Moreover, the function $S^{\mathfrak{i}}=\|\bar{d} R^{\mathfrak{i}}\|^2_{h}$ on $\E^1 \subset J^1 \pi$ is a relative invariant with weight $4 \lambda$. 
\end{theorem}
\begin{proof} 
Let $X \in \g$ be a general element, and $X^{(0)}$ its prolongation to $\E^0$. The expression $L_X(K_1 \wedge \cdots \wedge K_r)$ is by the Leibniz rule a sum of exterior products that have factors $K_1 ,\dots, K_{i-1}, L_X K_i, K_{i+1}, \dots, K_r$ where $i$ runs from $1$ to $r$. Since $\mathfrak{i}$ is an ideal of $\g$, we have $L_X K_i= C_i^j K_j$ for some set of constants $C_i^j$. Thus the only terms that remain in this sum are constant multiples of $K_1 \wedge \dots \wedge K_r$, which proves the first part of the proposition. 

By denoting the norm of the $r$-form as $$||K_1^{(0)} \wedge \cdots \wedge K_r^{(0)}||_h^2 = h(K_1^{(0)} \wedge \cdots \wedge K_r^{(0)},K_1^{(0)} \wedge \cdots \wedge K_r^{(0)}),$$ since $L_{X^{(0)}} h = 0$ for every $X \in \g$, it follows that 
  \begin{align*} L_{X^{(0)}} R^{\mathfrak{i}} &= L_{X^{(0)}} (h(K_1^{(0)} \wedge \cdots \wedge K_r^{(0)},K_1^{(0)} \wedge \cdots \wedge K_r^{(0)})) \\&= 2 h(L_{X^{(0)}} (K_1^{(0)} \wedge \cdots \wedge K_r^{(0)}),K_1^{(0)} \wedge \cdots \wedge K_r^{(0)}) \\
  &= 2 \lambda_X h(K_1^{(0)} \wedge \cdots \wedge K_r^{(0)},K_1^{(0)} \wedge \cdots \wedge K_r^{(0)})\\
  &= 2 \lambda_X R^{\mathfrak{i}}, \qquad \forall X \in \g.
  \end{align*}
  Next, since the Lie-derivative with respect to $X^{(1)}$ commutes with the horizontal exterior derivative $\bar{d}$, we have 
  \[L_{X^{(1)}}(\bar{d}R^{\mathfrak{i}}) = \bar{d}(2\lambda_X R^{\mathfrak{i}})=2 \bar{d}(\lambda_X)R^{\mathfrak{i}}+2\lambda_X \bar{d}R^{\mathfrak{i}} = 2 \lambda_X \bar{d}R^{\mathfrak{i}},\]
  where the last equality follows from the fact that $\lambda_X$ is constant for every $X \in \mathfrak g$. 
It follows that
  \[L_{X^{(1)}} (h^{-1}(\bar{d}R^{\mathfrak{i}}, \bar{d}R^{\mathfrak{i}})) = 2 h^{-1} (L_{X^{(1)}} \bar{d} R^{\mathfrak{i}}, \bar{d} R^{\mathfrak{i}}) = 4 \lambda_X h^{-1}(\bar{d} R^{\mathfrak{i}}, \bar{d} R^{\mathfrak{i}}).\] 
\end{proof}
Notice that the choice of basis of  $\mathfrak{i}$ above influences the relative invariants $R$ and $S$ only by a constant scalar factor. Although the statement is technically true in general, it is possible that $R^{\mathfrak{i}}$ is constant or even constantly zero. This happens for example always when $\dim {\mathfrak i} \geq \dim M$. 

Given a metric $g \in \mathrm{Sol}(\E)$ on $M$, we can compute the restriction of $R^{\mathfrak{i}}$ and $S^{\mathfrak{i}}$ (or any other function on $\E^i$) to $g$: 
\[R^{\mathfrak{i}}_g = R^{\mathfrak{i}} \circ g, \qquad S^{\mathfrak{i}}_g = S^{\mathfrak{i}} \circ j^1 g.\] 
The functions $R^{\mathfrak{i}}_g$ and $S^{\mathfrak{i}}_g$ are functions on $M$. 

Our main application of Theorem \ref{th:R} will be to $\g$-invariant ideals of $\killing$. If $\mathfrak{i}$ is an ideal of $\killing$, then the vector fields of $\mathfrak{i}$ span a (not necessarily regular) distribution  in $TM$ which is invariant with respect to the vector fields of $\killing$. If in addition $\mathfrak{i}$ is an ideal in $\g$, this distribution is $\g$-invariant. This happens in particular when $\mathfrak{i}$ is a characteristic ideal of $\killing$.
\begin{definition} 
An ideal $\mathfrak{i}$ in $\killing$ is called a characteristic ideal if it satisfies  $D(\mathfrak{i}) \subset \mathfrak{i}$  for any derivation $D \colon \killing \to \killing$.
\end{definition} 
Recall that a derivation of the Lie algebra $\killing$ is a linear map $D \colon \killing \to \killing$  satisfying $D[X,Y] = [DX,Y]+[X,DY]$ for every pair $X,Y \in \killing$. It follows that if $\g$ is a (possibly infinite-dimensional) Lie algebra that contains $\killing$ as an ideal, and $\mathfrak{i}$ is a characteristic ideal of $\killing$, then $\mathfrak{i}$ is an ideal of $\g$. This is due to the fact that for any $Z \in \g$, the operation $\mathrm{ad}_Z = [Z,\cdot]$ is a derivation of $\killing$. Thus any characteristic ideal of $\killing$ can play the role of $\mathfrak{i}$ in Theorem \ref{th:R}.

\subsection{An abelian ideal of $\killing$} \label{sect:abelianideal}
For any Lie algebra $\killing$, the Levi decomposition lets us write 
\[\killing = \mathfrak{s} \ltimes \mathfrak{r},\] 
where $\mathfrak{s}$ is a semisimple Lie algebra and $\mathfrak{r}$ the radical (largest solvable ideal) of $\killing$, and we will be using this decomposition throughout the paper, so that $\mathfrak{s}$ and $\mathfrak{r}$ will always denote the components of the Levi decomposition of $\killing$.

The derived sequence of the Lie algebra $\mathfrak{r}$ is defined as follows:
\begin{equation}
    \mathfrak{r}_0 = \mathfrak{r}, \qquad \mathfrak{r}_i = [\mathfrak{r}_{i-1}, \mathfrak{r}_{i-1}]. 
\end{equation}
\noindent As $\mathfrak{r}$ is solvable, there exists some $k \in \mathbb{N}$ such that $\mathfrak{r}_k \neq 0 $ and $\mathfrak{r}_{k+1} = 0$. In this case $\mathfrak{r}_k$ is abelian. 
The radical $\mathfrak{r}= \mathfrak{r}_0$ is a characteristic ideal of $\killing$ by Theorem 2.5.13 of \cite{winter1972liealgebras}. Furthermore, if $\mathfrak{r}_i$ is a characteristic ideal, then $\mathfrak{r}_{i+1}$ is a characteristic ideal, since 
\[D \mathfrak{r}_{i+1} = D[\mathfrak{r}_i,\mathfrak{r}_i] \subset [D \mathfrak{r}_i,\mathfrak{r}_i]+[\mathfrak{r}_i,D \mathfrak{r}_i] \subset [\mathfrak{r}_i,\mathfrak{r}_i] = \mathfrak{r}_{i+1}.\]
It follows by induction that all the Lie algebras $\mathfrak{r}_0, \dots, \mathfrak{r}_k$ are characteristic ideals of $\killing$. In particular we have the following statement:
\begin{lemma}\label{lem:rideal}
    Let $\g$ be a Lie algebra and $\killing = \mathfrak{s} \ltimes \mathfrak{r}$ an ideal of $\g$. If $\mathfrak{r}_0, \mathfrak{r}_1,\dots$ is the derived sequence of $\mathfrak{r}$, then $\mathfrak{r}_i$ is an ideal of $\g$ for each $i=0,1,\dots$. 
\end{lemma}

For any Lie algebra $\mathfrak{h}$, we let $\mathfrak{z}(\mathfrak{h})$ denote the center of $\mathfrak{h}$:
\[\mathfrak{z}(\mathfrak{h}) = \{X \in \mathfrak{h} \mid [X,Y]=0, \forall Y \in \mathfrak{h}\}. \]
In particular, since $\mathfrak{r}_k$ is abelian, we have $\mathfrak{z}(\mathfrak{r}_k)= \mathfrak{r}_k$. 
The Lie algebra $\mathfrak{z}(\mathfrak{r}_i)$ is a characteristic ideal of $\killing$, for each $i$, because $\mathfrak{r}_i$ is a characteristic ideal: 
\[ [D \mathfrak{z}(\mathfrak{r}_i), \mathfrak{r}_i] \subset D[\mathfrak{z}(\mathfrak{r}_i),\mathfrak{r}_i] + [\mathfrak{z}(\mathfrak{r}_i),D\mathfrak{r}_i] \subset [\mathfrak{z}(\mathfrak{r}_i),\mathfrak{r}_i] = \{0\}. \]
This argument was adapted from \cite{seligman1957characteristic}. It is obvious that the sum of characteristic ideals is a characteristic ideal, resulting in the following statement:
\begin{prop}\label{prop:aideal}
Let $\g$ be a Lie algebra and $\killing=\mathfrak{s} \ltimes \mathfrak{r}$ an ideal of $\g$. If $\mathfrak{r}_0, \mathfrak{r}_1, \dots, \mathfrak{r}_k$ is the derived sequence of $\mathfrak{r}$ and $\mathfrak{r}_{k+1}=\{0\}$, then the abelian Lie algebra $\mathfrak{a}(\mathfrak{r}) \subset \mathfrak{r}$ defined by 
\begin{equation}
    \mathfrak{a}(\mathfrak{r}) = \mathfrak{z}(\mathfrak{r}_0) + \mathfrak{z}(\mathfrak{r}_1) + \ldots + \mathfrak{z}(\mathfrak{r}_k). \label{eq:a(r)}
\end{equation}
is a characteristic ideal of $\killing$. In particular, $\mathfrak{a}(\mathfrak{r})$ is an ideal of $\g$. 
\end{prop}

It turns out that the abelian ideal $\mathfrak{a}(\mathfrak{r})$ intersects nontrivially with each ideal of $\mathfrak{r}$. This turns out to be a useful property that will be utilized in the next section. 
\begin{prop} \label{prop:intersection}
Let $\mathfrak{i} \neq \{0\}$ be an ideal of $\mathfrak{r}$. Then $\mathfrak{i} \cap \mathfrak{a}(\mathfrak{r}) \neq \{0\}$. 
\end{prop}
\begin{proof}
Let $\mathfrak{r}_k$ be the last nontrivial element of the derived sequence of $\mathfrak{r}$. If $\mathfrak{i} \cap \mathfrak{r}_k \neq \{0\}$, then it is obvious that $\mathfrak{i} \cap \mathfrak{a}(\mathfrak{r}) \neq \{0\}$. If $\mathfrak{i} \cap \mathfrak{r}_k = \{0\}$, then there exists an integer $j < k$ such that $\mathfrak{i} \cap \mathfrak{r}_j \neq \{0\}$ and $\mathfrak{i} \cap \mathfrak{r}_{j+1} = \{0\}$. 
We have 
\[[\mathfrak{i} \cap \mathfrak{r}_j, \mathfrak{r}_j] \subset \mathfrak{i} \cap \mathfrak{r}_{j+1} = \{0\},\]
which implies that $\mathfrak{i} \cap \mathfrak{r}_j \subset \mathfrak{z}(\mathfrak{r}_j) \subset \mathfrak{a}(\mathfrak{r}).$
\end{proof}

\subsection{Killing horizons} \label{sect:killingintro}
Theorem \ref{th:R} explains how to construct a relative invariant from an ideal of $\g$, and Proposition \ref{prop:aideal} gives us one particular such ideal, $\mathfrak{a}(\mathfrak{r})$. It turns out that $\mathfrak{a}(\mathfrak{r})$ can be used to construct a relative differential invariant that always vanishes on $\killing$-invariant Killing horizons. This will be the focus of this section.

We use the following definition, adapted from \cite{mars2018multiple}:
\begin{definition}
    A Killing horizon on a Lorentzian manifold $(M,g)$ is a null hypersurface $\mathcal H \subset M$ for which there exists a Killing vector $K \in \killing$ which satisfies $g(K,K)|_{\mathcal H}=0$ and $K_a \in T_a \mathcal H \setminus \{0\}$ for each $a \in \mathcal{H}$. 
\end{definition}
The Killing vector $K$ in the definition is called a generator of the Killing horizon, and it satisfies $K_a^\perp = T_a \mathcal H$ for each $a \in \mathcal H$. More generally, we will call $K$ a generator of the Killing horizon $\mathcal{H}$ if $g(K,K)|_{\mathcal H}=0$ and $K_a \in T_a \mathcal H \setminus \{0\}$ for each $a$ in an open dense subset of $\mathcal{H}$. Let $\generators$ denote the union of all generators of $\mathcal{H}$ and the trivial Killing vector $0 \in \killing$. By Theorem 2 of \cite{mars2018multiple}, the set $\generators$ is a Lie subalgebra of $\killing$. If $\dim(\generators) >1$, then $\mathcal{H}$ is called a multiple Killing horizon. 
\begin{remark}
   We will not worry about topological properties of $\mathcal{H}$, for example whether it is an embedded or injectively immersed hypersurface. 
\end{remark}

In order to focus in on Killing horizons that are geometrically special, we introduce the following definition: 

\begin{definition}
    Let $(M,g)$ be a Lorentzian manifold, and $\killing$ the Lie algebra of its Killing vector fields. A Killing horizon $\mathcal{H} \subset M$ is $\killing$-invariant if each vector-field in $\killing$ is tangent to $\mathcal{H}$. That is,  $X_a \in T_a \mathcal{H}$ for each $X \in \killing$ and $a\in \mathcal{H}$.
\end{definition}

To motivate our focus on $\killing$-invariant Killing horizons, notice that if $\mathcal{H}$ is not $\killing$-invariant, then there exists a vector field $X \in \killing$ satisfying $X_a \notin T_a \mathcal{H}$ for some $a \in \mathcal{H}$.  In this case, the flow $\varphi_s$ of $X$ gives a continuous family $\mathcal{H}_s=\varphi_s^{-1}(\mathcal H)$ of distinct Killing horizons, at least in a neighborhood of $a$ in $M$. Since $\varphi_s$ is an isometry on $(M,g)$ for each $s$, there is no way to pick out a single, special Killing horizon using local arguments, as they are all geometrically (locally) equivalent. There are several well-known spacetimes that are foliated by Killing horizons in this way, the simplest example being the 2-dimensional Minkowski space $(\mathbb R^2, -dt^2+dx^2)$, with Killing algebra $\langle \partial_t, \partial_x, x\partial_t+t\partial_x\rangle$, and Killing horizons generated by $(x-x_0)\partial_t+(t-t_0) \partial_x$ for each $(t_0,x_0) \in \mathbb R^2$. 

We  start with a simple result concerning $\killing$-invariant Killing horizons. 
\begin{lemma} \label{lem:generator}
  Let $(M,g)$ be a Lorentzian manifold, and $\killing$ the Lie algebra of its Killing vector fields. Assume that there exists a Killing horizon $\mathcal{H} \subset M$ with generator $K$. 
  Then $\mathcal{H}$ is $\killing$-invariant if and only if $g(X,K)|_{\mathcal H}\equiv 0$ for every $X \in \killing$. 
 
\end{lemma}
\begin{proof}
    The Killing horizon $\mathcal{H}$ is $\killing$-invariant if and only if $X_a \in T_a \mathcal H$ for every $X \in \killing$ and every $a \in \mathcal H$. Since $K$ is a generator of $\mathcal{H}$, we have $K_a \in (T_a \mathcal{H})^\perp$ for each $a \in \mathcal H$, implying that $g(X,K)|_{\mathcal{H}}\equiv 0$ for every $X \in \killing$. Conversely, if $g(X,K)|_{\mathcal{H}} \equiv 0$, then $X_a \in K_a^\perp$ for each $a \in \mathcal{H}$, implying that $X_a \in T_a \mathcal{H}$ for each $a$ in an open dense subset of $\mathcal{H}$. By continuity of the vector field $X$, $\mathcal{H}$ is $\killing$-invariant. 
\end{proof}

The following proposition is a consequence of Lemma \ref{lem:generator}.

\begin{prop}\label{prop:Rig}
    Let $(M,g)$ be a Lorentzian manifold, $\killing$ its Lie algebra of Killing vectors, and assume that there exists a $\killing$-invariant Killing horizon $\mathcal{H} \subset M$. Let $\mathfrak i = \langle K_1, \dots, K_r\rangle$ be an ideal of $\g$ satisfying $\mathfrak{i} \cap \generators \neq \emptyset$. Then the function 
    \[R^{\mathfrak{i}}_g=\| K_1 \wedge \cdots \wedge K_r \|_g^2\] 
    vanishes on $\mathcal{H}$. Furthermore, $R^{\mathfrak{i}}_g$ is simply the restriction of the relative invariant of Theorem \ref{th:R} to $g$: $R^{\mathfrak{i}}_g = R^{\mathfrak{i}} \circ g$. 
\end{prop}
\begin{proof}
We assume without loss of generality that $K_1 \in \generators$. The quantity $g(K_1 \wedge \cdots \wedge K_r,K_1 \wedge \cdots \wedge K_r)$ can be expressed in terms of a sum of products of $g(K_i,K_j)$ where each term in the sum contains a factor of the form $g(K_1,K_j)$. Since each $X \in \killing$ is tangent to $\mathcal{H}$ (by $\killing$-invariance), and since $K_1$ is a generator of the null-hypersurface $\mathcal{H}$, it follows from Lemma \ref{lem:generator} that $g(K_1,X)|_{\mathcal{H}}=0$ for every $X \in \killing$. Thus \[(R^{\mathfrak{i}}_g)|_{\mathcal{H}}=g(K_1 \wedge \cdots \wedge K_r,K_1 \wedge \cdots \wedge K_r)|_{\mathcal{H}} =0.\]
The last sentence of the proposition follows directly from Theorem \ref{th:R} (and this is the only place we use that $\mathfrak{i}$ is $\g$-invariant). 
\end{proof}
\begin{remark}
    Our focus is on $\killing$-invariant Killing horizons for the reasons explained above. Still, it is worth pointing out that if $\mathfrak{i}$ is also an ideal of $\killing$, it is sufficient that $\mathcal{H}$ is $\mathfrak{i}$-invariant (not necessarily $\killing$-invariant) for the conclusion of Proposition \ref{prop:Rig} to hold. 
\end{remark}

 When $\mathcal{H}$ is a $\killing$-invariant Killing horizon, the Lie algebra $\generators$ is not only a Lie subalgebra, but also an ideal in $\killing$. Before we prove that, we recall the following lemma from Appendix A of \cite{mars2018multiple}. 
\begin{lemma}[\cite{mars2018multiple}] \label{lem:codim2}
    If $X \neq 0$ is a Killing vector of a Lorentzian manifold $(M,g)$, then the subset on which $X$ vanishes has codimension at least 2. 
\end{lemma}
This lemma turns out to be quite useful. For example, if $\mathcal{H} \subset M$ is a $\killing$-invariant hypersurface, then there exists a Lie algebra homomorphism
\[ \killing \to \killing|_{\mathcal{H}},\]   
defined by restricting the vector fields of $\killing$ to $\mathcal{H}$ ($\killing|_{\mathcal{H}}$ should not be confused with $\generators$). From Lemma \ref{lem:codim2} it follows that this is  a Lie algebra isomorphism.

\begin{prop}\label{prop:generatorsideal}
    Let $(M,g)$ be a Lorentzian manifold, and $\killing$ the Lie algebra of its Killing vector fields. For any $\killing$-invariant Killing horizon $\mathcal{H} \subset M$, the Lie algebra $\generators$ of generators of $\mathcal{H}$ forms an ideal in $\killing$.   
\end{prop}

\begin{proof}
 
 For $X \in \killing$ we have $L_X g=0$, implying that 
   \[L_X(g(Y,Z)) = g([X,Y],Z) + g(Y,[X,Z]), \qquad X,Y,Z \in \killing.\]
   If $K \in \generators$, then $g(K,Z)|_{\mathcal{H}}\equiv 0$ for every $Z \in \killing$, by Lemma \ref{lem:generator}. By setting $Y=K$ in the above equation, we see that
   \[g([X,K],Z)|_{\mathcal{H}}=L_X(g(K,Z))|_{\mathcal{H}}=0, \qquad X,Z \in \killing. \]
   The last equallity holds since $g(K,Z)|_{\mathcal{H}}=0$ and $X$ is tangent to $\mathcal{H}$. Thus $[X,K]$ is either a generator for $\mathcal{H}$ or it vanishes on $\mathcal{H}$ (in which case it vanishes everywhere, by Lemma \ref{lem:codim2}).   
\end{proof}

Theorem 3 of \cite{mars2018multiple} says that if $m = \dim \generators \geq 2$, then there exists an abelian ideal $\generatorsdeg \subset \generators$ of dimension at least $m-1$. Thus $\generators$ is solvable, and if $\mathcal{H}$ is $\killing$-invariant, then $\generators$ is a solvable ideal of $\killing$. It follows that $\generators$ is an ideal of the radical $\mathfrak{r}$ of $\killing = \mathfrak{s} \ltimes \mathfrak{r}$. Thus Proposition \ref{prop:generatorsideal} leads to the following corollary: 

\begin{cor}
   Let  $(M,g)$ be a Lorentzian manifold with Killing algebra $\killing$, and let $\mathfrak{s} \ltimes \mathfrak{r}$ be the Levi decomposition of $\killing$. If $\mathcal{H} \subset M$ is a $\killing$-invariant Killing horizon, then the ideal $\generators$ of $\killing$ is contained in $\mathfrak{r}$. In particular, $\killing$ has a nontrivial radical ideal. 
\end{cor}

Let us take the ideal $\mathfrak{a}(\mathfrak{r})$ from Proposition \ref{prop:aideal}. Since $\generators$ is an ideal of $\mathfrak{r}$, it follows from  Proposition \ref{prop:intersection} that $\generators \cap \mathfrak{a}(\mathfrak{r}) \neq \{0\}$. Proposition \ref{prop:Rig} then gives us a relative invariant that vanishes on $\mathcal{H}$. We summarize this in a theorem: 
\begin{theorem}\label{th:invariantH}
    Let $(M,g)$ be a Lorentzian manifold with Killing algebra $\killing = \mathfrak{s} \ltimes \mathfrak{r}$, let $\g$ be a Lie algebra of vector fields on $M$ containing $\killing$ as an ideal, and let $\mathfrak{a}(\mathfrak{r}) = \langle K_1,\dots, K_r\rangle$ be the abelian ideal defined by \eqref{eq:a(r)}. If $\mathcal{H}$ is a $\killing$-invariant Killing horizon, then the function 
    \[ R_g = \| K_1 \wedge \cdots \wedge K_r \|_g^2\]
    vanishes on $\mathcal{H}$. The function $R_g$ is the restriction of the relative invariant of Theorem \ref{th:R} to $g$: $R_g = R^{\mathfrak{a}(\mathfrak{r})} \circ g$. 
\end{theorem}

We have thus shown that the relative differential invariant $R^{\mathfrak{a}(\mathfrak{r})}$
defined on $\E^0 \subset J^0 \pi$ detects $\killing$-invariant Killing horizons, in the sense that the function $R_g = R \circ g \colon M \to \mathbb R$ vanishes on them. 

Note that while $R_g$ always vanishes on $\killing$-invariant Killing horizons, it is possible that a hypersurface given by $R_g=0$ is not a Killing horizon. An example of this can be seen in Section \ref{sect:eichhorn}. It is also possible in general that $\{R_g=0\} \subset M$ is not a hypersurface at all. We end this section by looking at two examples where $R_g$ vanishes everywhere on $M$, for two different reasons.

\begin{example}
    The pp-wave 
    \[ g= dx^2+dy^2+(2 dv + H(x,y,u) du) du\]
    on $\mathbb R^4$ has Killing algebra $\killing = \langle \partial_v \rangle$ for generic $H$. In this case $\mathfrak{a}(\mathfrak{r}) = \killing$, and the function $R_g= \|\partial_v \|_g^2$ vanishes everywhere. This can be explained by the fact that $(\mathbb R^4,g)$ is foliated by $\killing$-invariant Killing horizons. They are the level sets of the coordinate function $u$, and they have generator $\partial_v$. 
\end{example}

\begin{example}
    The particular highly symmetric pp-wave (see \cite[p.~194]{kruglikovthe2014gap})
    \[ g = dx^2+dy^2+(2dv+x^2 du)du\]
    on $\mathbb R^4$ has the six-dimensional Killing algebra 
    \[ \killing = \langle \partial_y, \partial_u, \partial_v, u \partial_y-y \partial_v, e^{-u}(\partial_x+x\partial_v), e^{u}(\partial_x-x\partial_v) \rangle.\]
    The abelian ideal $\mathfrak{a}(\mathfrak{r}) = \langle \partial_y,\partial_v \rangle$ gives $R_g=\| \partial_y \wedge \partial_v\|_g^2$. Since $\killing$ is transitive, there are no $\killing$-invariant Killing horizons. Still, it is clear that $R_g$ vanishes everywhere. This is because $\partial_y$ and $\partial_v$ are both tangent to the Killing horizons $u=\mathrm{const}$. 
\end{example}

\section{Analysis of orbits on $\E^0$ and $\E^1$}
\label{sect:examples}
In this section we collect examples of computation. We look at cases where $\killing$ is isomorphic to the Lie algebras
\[\mathfrak{so}(3) \oplus \mathbb R, \quad  \mathbb R^2,\quad  \mathfrak{sol}(2), \quad \mathfrak{so}(3),\] 
which is the case for several well-known metrics. For each of these cases, we will compute the Lie algebra $\g$ of vector fields preserving the coordinate form of $\killing$, and a PDE $\mathcal{E}$ whose solutions have Killing algebras containing $\killing$. Then we will analyse the orbits on $\E^0$, and in some cases also on $\E^1$. From this analysis we find a relative invariant that vanishes on physically important hypersurfaces. In the first example we will be quite detailed in order to clearly explain how the computations are done, while the later examples we will be more brief. 

When $\killing$ has a nontrivial ideal, in all the examples we examine, we will see that the relative invariant $R^{\mathfrak{a}(\mathfrak{r})}$  vanishes on the $\g^{(0)}$-invariant hypersurface in $\E^0 \subset J^0 \pi$ containing orbits of submaximal dimension. 
\subsection{Static and spherically symmetric ($\killing = \mathfrak{so}(3) \oplus \mathbb R$)} \label{sect:Schwarzschild}
Consider the Lie algebra 
\[ \killing = \left\langle \partial_t, \partial_{\varphi}, \sin(\varphi) \partial_\theta+\frac{\cos(\varphi)}{\tan(\theta)} \partial_\varphi,\cos(\varphi) \partial_\theta-\frac{\sin(\varphi)}{\tan(\theta)} \partial_\varphi\right\rangle\]
of vector fields defined on $\mathbb R \times (0,\infty) \times (0,\pi) \times (0, 2\pi)$, where $t,r,\theta,\phi$ are the coordinates on each factor, respectively.
\begin{prop} \label{prop:Schwarzschild}
The Lie algebra $\g$ preserving the Lie algebra $\killing$ of Killing vectors is infinite-dimensional, and spanned by the vector fields 
\[ t \partial_t, \quad a(r) \partial_t, \quad b(r) \partial_r, \quad   a,b \in C_{\mathrm{loc}}^\infty((0,\infty))\]
in addition to the last three generators of $\killing$:
\[\partial_{\varphi}, \quad  \sin(\varphi) \partial_\theta+\frac{\cos(\varphi)}{\tan(\theta)} \partial_\varphi, \quad \cos(\varphi) \partial_\theta-\frac{\sin(\varphi)}{\tan(\theta)} \partial_\varphi.\]
\end{prop} 
Notice that the vector field $\partial_t$ is already included in $\g$, just set $a(r) \equiv 1$. A metric has Killing algebra containing $\killing$ if and only if it takes the form
\[ g = g_{11}(r) dt^2+g_{12}(r) dt dr+g_{22}(r) dr^2+g_{33}(r) (d\theta^2+\sin^2(\theta)d\varphi^2),\] 
where $g_{11}, g_{12}, g_{22}, g_{33}$ are functions on $(0,\infty)$. The PDE $\mathcal{E}$ determining $g$ is given by 
\begin{align*} 
 \mathcal{E}^0 &= \{u_{13}=0,u_{14}=0,u_{23}=0, u_{24}=0,u_{34}=0, u_{44} = \sin^2(\theta) u_{33}\}
 \end{align*}
on $J^0\pi$. The differential constraints of $\mathcal{E}^1 \subset J^1 \pi$ are given by the total derivatives of the constraints of $\mathcal{E}^0$ and the additional constraints $(u_{ij})_t=0,(u_{ij})_\theta=0, (u_{ij})_\varphi=0$. For each $k\geq0$ we have $\dim \E^k = 8+4k$. The restriction of the horizontal symmetric 2-form $h$ from Section \ref{sect:jets} is given by 
\[ h= u_{11} dt^2 + u_{12} dt dr+u_{22} dr^2 + u_{33} (d\theta^2 + \sin^2(\theta) d \varphi^2).\]
Each element $X \in \g$ can be lifted to a vector field $X^{(0)}$ on $\E^0$, by the requirement that $L_{X^{(0)}} h=0$. The lifts of the first three vector-fields in Proposition \ref{prop:Schwarzschild}  are given by 
\[ t \partial_t-2u_{11} \partial_{u_{11}}-u_{12} \partial_{u_{12}}, \; a \partial_t-2 a' u_{11} \partial_{u_{12}}-a' u_{12} \partial_{u_{22}}, \; b \partial_r-b'u_{12}\partial_{u_{12}}-2 b' u_{22} \partial_{u_{22}},\]
respectively, while the lifts of the elements in $\killing$ are trivial (they have no vertical components). 

As ${\mathfrak{a}(\mathfrak{r})} = \langle \partial_t \rangle$, the relative differential invariant of Theorem \ref{th:R} is given by  
 \[R^{\mathfrak{a}(\mathfrak{r})}=\|\partial_t\|_h^2=u_{11}.\]  
 \begin{prop}
Generic $\g^{(0)}$-orbits on $\E^0$ are 7-dimensional. The field of absolute invariants on $\E^0$ is generated by the absolute invariant $u_{33}$. The subset of $\E^0$ on which the orbit dimension is less than 7 is given exactly by $R^{\mathfrak{a}(\mathfrak{r})}=0$. All orbits on this 7-dimensional subset are 6-dimensional. 
 \end{prop} 
 \begin{proof}
   It is easy to verify that $u_{33}$ is an absolute invariant. Since $\dim \E^0=8$ this implies that the dimension of $\g^{(0)}$-orbits is at most 7. To verify that this is upper bound is attained, we look at the lift of the 8-dimensional Lie subalgebra of $\g$ spanned by
   \[ \left\langle t \partial_t, \partial_t, r \partial_t, \partial_r, r\partial_r,\partial_{\varphi},   \sin(\varphi) \partial_\theta+\frac{\cos(\varphi)}{\tan(\theta)} \partial_\varphi,  \cos(\varphi) \partial_\theta-\frac{\sin(\varphi)}{\tan(\theta)} \partial_\varphi  \right\rangle.\]
Lining the generators up in an $8 \times 8$ matrix and computing the rank, shows that the generic rank is $7$, implying that the upper bound is attained by generic orbits. The rank of this matrix drops when $u_{11} u_{22} - 4 u_{12}^2=0$ or $u_{11} =0$. (These computations are well suited for computer algebra systems. We have used Maple with its PolynomialIdeal package.) The first of these equations never holds since we assume that $h$ is nondegenerate. The second equation is exactly $R^{\mathfrak{a}(\mathfrak{r})}=0$. To see that the same result is true for the infinite-dimensional Lie algebra $\g^{(0)}$, we look at the expressions for the lifted Lie algebras. When $u_{11}=0$, the vector fields of $\g^{(0)}$ have vanishing $\partial_{u_{11}}$-component and vanishing $\partial_{u_{33}}$-component, implying that the rank on the subset given by $u_{11}=0$ is never greater than 6. It is also easy to check that the rank of the $8\times 8$ matrix never drops below 6 as long as $h$ is nondegenerate. 
 \end{proof}
Notice that the absolute invariant $u_{33}$ can also be described in terms of the Lie algebra $\killing$:
\[\|\partial_{\varphi}\|_h^2 +\left\|\sin(\varphi) \partial_\theta+\frac{\cos(\varphi)}{\tan(\theta)} \partial_\varphi\right\|_h^2 + \left\|\cos(\varphi) \partial_\theta-\frac{\sin(\varphi)}{\tan(\theta)} \partial_\varphi\right\|_h^2 = 2 u_{33}.\] 
More generally, let $Y_1,Y_2,Y_3$ be a basis of the ideal $\mathfrak{so}(3)$ of $\killing$ which is orthonormal with respect to the Killing form on $\mathfrak{so}(3)$. Then $\|Y_1\|_h^2 +\|Y_2\|_h^2 +\|Y_3\|_h^2$ is proportional to $u_{33}$.

\subsubsection{The Reissner-Nordström metric}
As an example, for the  Reissner-Nordström metric 
\begin{align*}
g=-\left(1-\frac{r_s}{r} +\frac{r_Q^2}{r^2}\right) c^2 dt^2+\left(1-\frac{r_s}{r} +\frac{r_Q^2}{r^2}\right)^{-1} dr^2 +r^2 (d\theta^2+\sin^2 \theta d\phi^2)
\end{align*}
 we have \[R^{\mathfrak{a}(\mathfrak{r})}_g=\|\partial_t\|_g^2=-\left(1-\frac{r_s}{r} +\frac{r_Q^2}{r^2}\right) c^2.\] 
This function vanishes exactly on the event horizon, which is a $\killing$-invariant Killing horizon with generator $\partial_t$. For the Schwarzschild metric, we have $r_Q=0$ and $R^{\mathfrak{a}(\mathfrak{r})}_g=- (1-r_s/r)c^2$.

\subsection{Stationary and axisymmetric ($\killing=\mathbb R^2$)} \label{sect:R2}
Consider the 2-dimensional abelian Lie algebra \[ \killing = \left\langle \partial_{x^0}, \partial_{x^1}\right\rangle\] 
on $M=\mathbb R^4(x^0,x^1,x^2,x^3)$. 
\begin{prop} \label{prop:gex2}
The Lie algebra $\g$ of vector fields preserving the Lie algebra $\killing$ of Killing vectors is infinite-dimensional and spanned by the vector fields 
\[ x^0 \partial_{x^0},\quad x^0 \partial_{x^1},\quad x^1 \partial_{x^0},\quad x^1 \partial_{x^1},\quad  \sum_{i=0}^3 a^i(x^2, x^3) \partial_{x^i}, \qquad a^i \in C_{\mathrm{loc}}^\infty(\mathbb R^2).\]
\end{prop}
The general metric admitting these Killing vectors is given by 
\[
g = \sum_{i\leq j} g_{ij}(x^2, x^3) dx^i dx^j.\]
It can be thought of as the solution to the PDE system defined by 
\[\E^1 = \{(u_{ij})_{x^0} = 0, (u_{ij})_{x^1}=0 \mid 0 \leq i  \leq j \leq 3\}. \]
In this case $\E^0 = S_{\mathrm{Lor}}^2 T^* M$. The horizontal symmetric 2-form $h$ is given by 
\[ h= \sum_{i \leq j} u_{ij} dx^i dx^j.\]
Since $\killing$ is abelian, there is a priori no special Killing vectors. In fact, the Lie pseudogroup corresponding to $\g$ acts transitively on the 2-dimensional space of Killing vectors. We have $\mathfrak{a}(\mathfrak{r}) = \killing$, and the relative invariant $R^{\mathfrak{a}(\mathfrak{r})}$ of Theorem \ref{th:R} is given by 
\[ R^{\mathfrak{a}(\mathfrak{r})}=\|\partial_{x^0} \wedge \partial_{x^1}\|_h^2 = 2 u_{00}u_{11}-u_{01}^2/2. \]
\begin{prop}
    There is a (14-dimensional) open $\g^{(0)}$-orbit on $\E^0$, and thus no absolute invariants on $\E^0$. The subset of $\E^0$ on which the orbit dimension is less than 14 is a reducible algebraic set, whose 13-dimensional component is given exactly by $R^{\mathfrak{a}(\mathfrak{r})}=0$. 
\end{prop}
\begin{proof}
    To show that there is a 14-dimensional open orbit on $\E^0$, take for example the 16 independent vector fields in $\g$ having polynomial coefficients of degree $\leq 1$, and lift them to $\E^0$. It is easily verified that the rank at a generic point is 14. Next, one can check that the rank of these 16 vector fields drops exactly on the algebraic subset $\mathcal{S}_{13} \cup \mathcal{S}_{11} \subset \E^0$ where $\mathcal{S}_{ 13} = \{R^{\mathfrak{a}(\mathfrak{r})}=0\}$ has  dimension 13 and $\mathcal{S}_{11}$ has  dimension 11 (we remind that $\det(h)$, which appears in these computations, never vanishes on $S^2_{\mathrm{Lor}} T^*M$). In the end we check that the set $\mathcal{S}_{13} \cup \mathcal{S}_{11}$ is $\g^{(0)}$-invariant (since we so far used only a 16-dimensional Lie subalgebra). 
\end{proof}
Thus, there is a unique $\g^{(0)}$-invariant hypersurface in $\E^0$, and it is given by $R^{\mathfrak{a}(\mathfrak{r})}=0$. 

We note that the analysis provided in this section is similar to that in \cite{ferraioli2020equivalence} where they provide a complete list of all absolute differential invariants for such spacetimes. However, our focus is on finding relative differential invariants of the lowest possible order that detect the Killing horizon.

\subsubsection{The Kerr-Newman metric}
As an example, consider the Kerr-Newman metric in coordinates $x^0=t,x^1=\varphi, x^2 = r, x^3 = \theta$  whose nonzero components are given by \cite{poisson2004}
\begin{gather*}
    g_{00} = -\frac{\Delta - a^2 \sin^2(\theta)}{\Sigma}, \qquad  g_{01} = -\frac{2a\sin^2(\theta) (r^2+a^2-\Delta)}{\Sigma} \\
    g_{11} = \frac{(r^2-a^2)^2-\Delta a^2 \sin^2(\theta)}{\Sigma} \sin^2(\theta), \qquad g_{22} = \frac{\Sigma}{\Delta}, \qquad g_{33} = \Sigma,
\end{gather*}
where $\Sigma = r^2 +a^2\cos^2(\theta)$ and $\Delta = r^2 + a^2 + e^2 -2Mr$. In this case, we have 
\[ R^{\mathfrak{a}(\mathfrak{r})}_g = -2\sin^2(\theta)(r^2-2Mr+(a^2+e^2)),\]
which vanishes exactly on the event horizon, which is well-known to be a Killing horizon. 

\subsubsection{A metric of Eichhorn and Held} \label{sect:eichhorn}
The invariant $R$ may also detect hypersurfaces that are not Killing horizons. Consider the stationary and axisymmetric spacetime given in \cite{eichhorn2021locality}: 

\begin{eqnarray}
   g &&= -\frac{r^2 - 2M r + a^2 \chi^2}{r^2 + a^2 \chi^2} du^2 + 2du dr - \frac{4 M ar}{r^2 + a^2 \chi^2}(1-\chi^2) du d\phi \nonumber \\
   && -2a(1-\chi^2) dr d\phi + \frac{r^2 + a^2 \chi^2}{1-\chi^2} d\chi^2 \\
   && + \frac{1-\chi^2}{r^2+a^2 \chi^2}\left( (a^2+r^2)^2 - a^2 (r^2 -2 Mr + a^2)(1-\chi^2) \right) d\phi^2. \nonumber 
\end{eqnarray}
Here $u,r,\phi,\chi=\cos(\theta)$ are coordinates and $M$ is a function of $r$ and $\chi$. The Lie algebra of Killing vectors is 2-dimensional: $\killing = \langle \partial_u,\partial_\phi\rangle$. The invariant $R$ restricted to $g$ takes the form 
\[R^{\mathfrak{a}(\mathfrak{r})}_g = \|\partial_\varphi \wedge \partial_u \|^2_{g}=2(1-\chi^2) (2M(r,\chi)r-a^2-r^2).\]
We have $\chi \in (-1,1)$ since $\theta \in (0, \pi)$. Therefore, the invariant vanishes if and only if  $f:=(2M(r,\chi)r-a^2-r^2) = 0$, and we use the notation $\mathcal{H} = \{f=0\}$. However, the surface $\mathcal{H}$ is not null, since the normal vector field
\begin{eqnarray}
    g^{-1}(df, \cdot)|_{\mathcal{H}} &&= \frac{2(a^2+r^2)(r \partial_r M + M-r) }{a^2\chi^2 + r^2 } \partial_u \nonumber \\ && - \frac{2( \chi^2 - 1) r \partial_{\chi} M}{a^2 \chi^2 + r^2} \partial_{\chi} + \frac{2a(r\partial_r M + M -r)}{a^2 \chi^2 + r^2} \partial_{\phi},
\end{eqnarray}
\noindent has, in general, non-vanishing norm:
\[ \| df \|_{g}^2|_{\mathcal{H}} = -\frac{4 (\chi^2 -1) r^2 (\partial_{\chi} M)^2}{a^2 \chi^2 +r^2}. \]

\subsection{Two-dimensional solvable ($\killing=\mathrm{sol}(2)$)}
Consider the Lie algebra 
\[ \killing = \langle \partial_v, v \partial_v-r\partial_r\rangle\]
on $\mathbb R^4(x^1,x^2,r,v)$. 
\begin{prop}
    The Lie algebra $\g$ of vector fields preserving the Lie algebra $\killing$ of Killing vectors are spanned by the vector fields 
    \[ \partial_v, \quad  v \partial_v-r \partial_r, \quad a_1(x) \partial_{x^1}+a_2(x)\partial_{x^2}+a_3(x)r \partial_r + a_4(x) r^{-1} \partial_v, \quad a^i \in C_{\mathrm{loc}}^\infty(\mathbb R^2).\]
\end{prop}

The general invariant metric takes the form 
\begin{equation}
\begin{aligned}
g &= \frac{g_{33}(x)}{r^2} dr^2+ \frac{g_{i3}(x)}{r}  dx^i dr + dv\left(g_{34}(x) dr+r g_{i4}(x) dx^i+r^2 g_{44}(x) dv\right) \\ & \quad + g_{ij}(x) dx^i dx^j, 
\end{aligned}
\end{equation}
where Einstein notation will be used and the indices satisfy $i \leq j$ in the last term. If we write the horizontal symmetric form as 
\[h= \frac{u_{33}}{r^2}  dr^2+\frac{u_{i3}}{r} dx^i dr + dv\left(u_{34} dr+r u_{i4} dx^i+r^2 u_{44} dv\right) + u_{ij} dx^i dx^j, \]
then the PDE $\E$ is given by the following first-order system and its derivatives:
\[ \E^1 = \{(u_{ij})_r=0, (u_{ij})_v=0\}.\]
Let us focus on a neighborhood of $r=0$, and consider the sub-PDE $\E_0^1 \subset \E^1$ given by the additional constraints $u_{13} = u_{23}=u_{33} =0$, and their derivatives (in particular $\E_0^0=\{u_{13} = u_{23}=u_{33} =0\}$). Its general solution has the form of $g$ above, but with $g_{13}, g_{23}, g_{33}$ identically equal to zero. This is the general expression of a Near Horizon geometry (see for example \cite[Sect. 4.5]{mars2018multiple}).

As we have considered the subclass of regular spacetimes, we will consider vector-fields that are regular around $r=0$ as well. Let $\g_0\subset \g$ be the Lie algebra of vector fields defined around $r=0$, i.e., those with $a_4 \equiv 0$. 

From the two invariant ideals, $\mathfrak{a}(\mathfrak{r}) = \langle \partial_v \rangle$ and $\killing$ itself, we get from Theorem \ref{th:R} two relative differential invariants of order $0$: 
\[ R^{\mathfrak{a}(\mathfrak{r})}|_{\E^0_0}= \|\partial_v\|_h^2 = r^2 u_{44}, \qquad R^{\killing}|_{\E_0^0} = \| \partial_v \wedge (v\partial_v-r\partial_r)\|_h^2 = -\frac{r^2 u_{34}^2}{2}.\]
Due to assumed nondegeneracy, $u_{34}$ never vanishes. It is easily verified that the two relative invariants have the same weight, implying that their ratio $u_{44}/u_{34}^2$ is an absolute differential invariant. 
\begin{prop}
    Generic $\g_0^{(0)}$-orbits on $\E_0^0$ are 10-dimensional. The field of rational absolute invariants on $\E^0$ is generated by the absolute invariant $u_{44}/u_{34}^2$. The orbit dimension drops exactly on the set $\{r=0\} \subset \E_0^0$. 
\end{prop}

\subsubsection{Near Horizon geometries}
For the Near Horizon geometries, the hypersurface $\mathcal{H}=\{r =0\}$ is well-known to be a Killing horizon \cite{pawlowski2004spacetimes}. In fact, it is a multiple Killing horizon since $\dim(\generators) =2$ (see \cite{mars2018multiple}):
\[ \|(v-v_0) \partial_v-r\partial_r\|_g^2 = r(v-v_0)(g_{44}(x) r(v-v_0)-g_{34}(x)).\]
This computation also shows that the spacetime is foliated by Killing horizons, given by $v=v_0$. The latter ones are not $\killing$-invariant, because $\partial_v \in \killing$, which explains why they are not detected by $R^{\mathfrak{a}(\mathfrak{r})}|_{\E_0^0}$. On the other hand, the restriction of $R^{\mathfrak{a}(\mathfrak{r})}|_{\E_0^0}$ to a  Near Horizon metric does vanish on $\mathcal{H}$, consistent with Theorem \ref{th:invariantH}.

\subsection{Spherical symmetry ($\killing=\mathfrak{so}(3)$)} \label{sect:spherical}
Consider the Lie algebra 
\[ \killing = \left\langle \partial_{\varphi}, \sin(\varphi) \partial_\theta+\frac{\cos(\varphi)}{\tan(\theta)} \partial_\varphi,\cos(\varphi) \partial_\theta-\frac{\sin(\varphi)}{\tan(\theta)} \partial_\varphi\right\rangle\]
defined on the open chart $\mathbb R \times (0,\infty) \times (0,\pi) \times(0,2\pi)$ with coordinates $t,r,\theta,\varphi$. This Lie algebra is abstractly $\mathfrak{so}(3,\mathbb R)$. Since $\killing$ is simple, the only ideal that can be used in the context of Theorem \ref{th:R} is $\killing$ itself. However, since $\killing$ is 3-dimensional while the distribution it spans in $T\E^0$ is 2-dimensional, the function $R^\killing$ is just identically vanishing. However, we can still find relative invariants by analyzing orbits on $J^0 \pi$ and $J^1 \pi$. 
\begin{prop}
    The Lie algebra $\g$ of vector fields preserving $\killing$ is given by 
\[ \g= \langle a(t,r) \partial_t + b(t,r)\partial_r \mid a,b \in C_{\mathrm{loc}}^\infty(\mathbb R \times (0,\infty)) \rangle \ltimes \killing.\] 
\end{prop} 

The general metric admitting these Killing vectors is given by 
\[ g = g_{11}(t,r) dt^2+g_{12}(t,r) dt dr+g_{22}(t,r) dr^2+g_{33}(t,r) (d\theta^2+\sin^2(\theta)d\varphi^2).\] 
Writing the horizontal symmetric form $h$ as
\[ h= u_{11} dt^2 + u_{12} dt dr+u_{22} dr^2 + u_{33} (d\theta^2 + \sin^2(\theta) d \varphi^2),\]
the PDE $\E$ determining $g$ is given by 
\[\E^0 = \{u_{13}=0,u_{14}=0,u_{23}=0,u_{24}=0,u_{34}=0,u_{44} = \sin^2(\theta) u_{3,3}\}\]
on $J^0 \pi$. The differential constraints of $\E^1 \subset J^1 \pi$ are given by the total derivatives of the constraints of $\E^0$, in addition to the constraints $(u_{ij})_\theta = 0, (u_{ij})_\varphi=0$. In particular $\dim \E^0 = 8$ and $\dim \E^1 = 16$. 

\begin{prop}
    All $\g^{(0)}$-orbits on $\E^0$ are $7$-dimensional. Generic $\g^{(1)}$-orbits on $\E^1$ are 14-dimensional, and the orbit dimension drops on the subset given by $(u_{33})_t=0, (u_{33})_r=0$. The field of absolute invariants on $\E^1$ is generated by the two invariants
    \[ I=u_{33}, \qquad J=\|\bar{d}I\|_h^2 = \frac{4\left(u_{11} (u_{33})_r^2-u_{12} (u_{33})_r (u_{33})_t+u_{22} (u_{33})_t^2\right)}{4u_{11} u_{22}-u_{12}^2}. \]
\end{prop}
Notice that the absolute invariant $I$ can be described in terms of the Lie algebra $\killing$ in the  same way as was done in Section \ref{sect:Schwarzschild}.

Among the invariant hypersurfaces that can be singled out by a function of the form $f(I,J)=0$, there is a special (irreducible) one, namely the one containing the $\g^{(1)}$-invariant subset $\{(u_{33})_t=0, (u_{33})_r=0\} \subset \E^1$. It is given by $J=0$ or, equivalently, by the vanishing of the relative invariant 
\[ Q =u_{11} (u_{33})_r^2-u_{12} (u_{33})_r (u_{33})_t+u_{22} (u_{33})_t^2.\]

\subsubsection{The imploding spherically symmetric metric}
Let us consider the imploding spherically symmetric metric in advanced coordinates (see \cite{ColeyMcNutt}):
\[ g= -2e^{\beta(t,r)}\left(1-\frac{2m(t,r)}{r}\right) dt^2+2e^{\beta(t,r)} dt dr+r^2  (d\theta^2+\sin^2(\theta)d\varphi^2). \]
We have $Q_g = -8 e^{\beta(t,r)} (r-2m(t,r)) r$, which vanishes on the ``future outer trapping horizon'' given by $r=2m(t,r)$. In general this horizon will not be a null hypersurface, except when $m_{,t} = 0$ when it becomes an isolated horizon. Otherwise it will be a spacelike or timelike hypersurface  \cite{bengtsson2011region}. 

\section{Discussion}

In general relativity it is helpful to determine invariantly defined hypersurfaces. These hypersurfaces give physical insight into the nature of solutions. An important example of such a hypersurface is the boundary of a black hole spacetime. For most black hole spacetimes, determining this boundary can be difficult. However, for idealized black hole spacetimes which admit a high degree of symmetry, encoded in terms of a Lie algebra $\killing$ of Killing vector fields, the boundary of the black hole spacetime is a null hypersurface known as a Killing horizon where a Killing vector field becomes null and acts as a generator of the null hypersurface. Such hypersurfaces can be characterized invariantly using the norm of the Killing vector field and using curvature invariants \cite{PageShoom2015, brooks2018cartan}.

While it is reasonable to suspect that for a given class of black hole solutions, the horizon can be detected by a relative differential invariant, such as with the stationary axisymmetric black hole solutions contained in \cite{marvan2008local, ferraioli2020equivalence}, it is less obvious that such invariants can be singled out in a systematic way without a priori knowledge of the location of the horizon. 

In this paper we focused on the Lie algebra $\g$ preserving a fixed, but general, finite-dimensional Lie algebra $\killing$ of Killing vectors of a family of spacetimes, and computed relative invariants with respect to the prolongation of $\g$ on appropriate jet bundles. We showed in Theorem \ref{th:R} that any finite-dimensional ideal of $\g$, $\mathfrak{i} = \langle K_1, \ldots, K_r \rangle$, gives rise to a relative differential invariant of order $0$: $$R^{\mathfrak{i}}= \| K_1^{(0)} \wedge \cdots \wedge K_r^{(0)}\|_{h}^2.$$ 

To properly model black hole horizons, which are located in a specific region of a spacetime, we considered Killing horizons that are invariant under the group of isometries, called $\killing$-invariant Killing horizons. We showed that there exists a particular $\g$-invariant abelian ideal $\mathfrak{i} = \mathfrak{a}(\mathfrak{r})$ of $\killing$. Theorem \ref{th:invariantH} guarantees that the corresponding relative differential invariant $R^{\mathfrak{a}(\mathfrak{r})}$ always vanishes on $\killing$-invariant Killing horizons. 

For several concrete examples of $\killing$, we also directly analyzed the $\g^{(0)}$-orbits in $J^0 \pi$ and $\g^{(1)}$-orbits in $J^1 \pi$ to successfully produce relative differential invariants of order $0$ or $1$ that vanish on horizons of several well-known black hole spacetimes. The obtained relative invariants are compared with $R^{\mathfrak{a}(\mathfrak{r})}$ in the cases where it makes sense. While this second approach is computationally more cumbersome than the first, it is also more general. For example, when the Killing algebra is $\mathfrak{so}(3)$ there are no $\killing$-invariant Killing horizons, but there is a relative differential invariant of order $1$ that detects the unique spherically symmetric apparent horizon of imploding spherically symmetric metrics.

This last example motivates the investigation of horizon detecting relative differential invariants for more general black hole solutions. The methods used in this paper could be extended to study conformal Killing horizons \cite{dyer1979conformal} by examining conformal Killing algebras. This could be applied to black hole solutions conformal to stationary black holes and give further insight into conformal Killing horizons as a valid black hole boundary for dynamical black holes \cite{mcnutt2017curvature, sherif2024existence}. More generally, for explicit classes of dynamical black hole solutions, such as the Robinson-Trautmann class of solutions \cite{podolsky2009past} or black hole spacetimes within the class of LRS spacetimes \cite{sherif2021geometric}, this approach may be able to provide insight on the appropriate boundary.

 \begin{section}*{Acknowledgments}
We are grateful to Henrik Winther, and Jerzy Lewandowski for thoughtful discussions. DM and EH were supported by the Norwegian Financial Mechanism 2014-2021 (project registration number 2019/34/H/ST1/00636), and the Tromsø Research Foundation (project “Pure Mathematics in Norway”).
 \end{section}

\bibliographystyle{unsrt-phys}
\bibliography{ref}

\end{document}